\documentclass[reqno,keywordsasfootnote,10pt]{article}
\usepackage[english]{babel}
\usepackage[latin1]{inputenc}
\usepackage{amsmath,amsthm,amsfonts,amssymb,times}
\numberwithin{equation}{section}
\usepackage[final]{epsfig}
\usepackage{color}

\newtheorem{theorem}{Theorem}[section]
\newtheorem{proposition}[theorem]{Proposition}
\newtheorem{lemma}[theorem]{Lemma}

\theoremstyle{definition}
\newtheorem{definition}{Definition}[section]

\newcommand{\E}{\mathbb{E}}

\newcommand{\R}{\mathbb{R}}

\def\E#1{\mathbb{E} \left( #1 \right)}  
\def\P#1{\mathbb{P} \left( #1 \right)} 
\def\g[#1,#2,#3]{\langle #1|\frac{1}{#2} |#3\rangle} 
  
\makeatletter

\DeclareMathOperator{\indfct}{\rm 1}

\DeclareMathOperator{\Tr}{\rm Tr}

\makeatother
\def\be{\begin{equation} }
\def\ee{\end{equation} }
\begin{document}
\date{\small October 4, 2008}

\title{On the Joint Distribution of Energy Levels \\  of Random Schr\"odinger Operators}   

\author{Michael Aizenman and Simone Warzel\footnote{Present address: Zentrum Mathematik, TU M\"unchen, Germany} \\ \small Departments of
Mathematics and Physics,\\ \small Princeton University, Princeton NJ
08544, USA.}
%

\maketitle

\begin{abstract} 
We consider operators with  random potentials on graphs, such as the lattice version of the random Schr\"odinger operator. 
The main result  is  a general bound on the probabilities of simultaneous occurrence of eigenvalues in specified distinct  intervals, with the corresponding eigenfunctions being separately localized within prescribed regions.  The bound generalizes the Wegner estimate on the density of states.   The analysis proceeds through a new multiparameter spectral averaging principle.    \\[4ex] 

(2000 Mathematics Subject Classification: 47B80, 60K40) 
\footnote {E-mail address: aizenman@princeton.edu, swarzel@princeton.edu}\\
\vspace*{1cm}
\end{abstract}



\newpage

\section{Introduction}

In this note we present some general bounds for the joint 
distribution of eigenfunctions of operators with random potential, in the discrete setting.   After finding that a 
natural multi-level extension of the Wegner bound on the density of states is not generally valid, we present a corrected version which is.   It consists of a   general bound on the probabilities of simultaneous occurrence of eigenvalues in distinct  intervals, with the corresponding eigenfunctions being separately localized within prescribed regions, in a sense made precise below.   The bound is derived through a suitable multi-parameter extension of the spectral averaging principle which is a familiar, and useful, element of the mathematical theory of Anderson  localization.  

Wegner-type bounds are of relevance for the analysis of  the extensions of the Schr\"odinger evolution to non-linear time evolutions and to  interactive extensions of the one-particle model.    
These are not discussed here, but let us note that such systems continue to attract attention, with interesting results presented  in ref.~\cite{AF88,BW07,FKS}  as well as  in a number of works which are currently in progress~\cite{FKS,CS}.

More explicitly, we consider  random operators acting in  the Hilbert space $ \ell^2(\Lambda) $, with $\Lambda$  a finite set,  of the form
\begin{equation} \label{H}  
 H_\Lambda(\omega)  = T + V(\omega)
\end{equation} 
where $ T $ is a fixed hermitian operator 
and the randomness, represented by  $ \omega$, 
enters only through a diagonal matrix $  V(\omega) = {\rm diag}(V_x(\omega))_{x \in \Lambda}  $.  Here $ \omega$  is a variable taking values in a probability space $(\Omega, \mathbb{P})$.   The joint distribution of $\{ V_x(\omega)\}_{x \in \Lambda} $ induces a probability measure on the space of realizations $\R^{|\Lambda|}$.  For convenience, and without loss of generality we identify $\Omega$ with this space, with $\{ V_x(\omega)\}$ given by the natural coordinates.  Expectation values with respect to the probability measure $ \mathbb{P} $ will be denoted by $\mathbb{E}$.  

By default, it will subsequently be assumed here that the joint distribution of the potential variables satisfies the following {\em regularity} condition:
\begin{description}
\item[Assumption R:] For each site $x\in \Lambda$, the conditional probability distribution of  $V_x$, conditioned on 
$ \{ V_y \}_{y\neq x} $, is absolutely continuous with respect to the Lebesgue measure,  and its density (i.e. the corresponding Radon-Nikodym derivative) is uniformly bounded by a constant, $\rho_\infty$.
\end{description}  

Among the general results which are known for such random operators, and  which have already played  useful roles in the mathematical analysis of Anderson localization and of the corresponding spectral statistics, are:
\begin{enumerate} 
\item {\bf Spectral simplicity}:  With probability one $ H_\Lambda(\omega) $ has only simple, i.e., non-degenerate, eigenvalues.  (A proof which does not rely on the more involved Minami estimate is spelled in the appendix.)
\item The {\bf Wegner bound}  \cite{Weg81}:  the mean density of states is bounded by $\rho_\infty$.  Equivalently, for any energy interval   $I$ 
  \begin{equation}  \label{wegner} 
   \mathbb{P}\left\{ \sigma(H_\Lambda) \cap I \neq \emptyset \right\} \ \le \ \mathbb{E}\left[\Tr \, P_{I}(H_\Lambda)\right] \ \le \ \rho_\infty \, |I| \, |\Lambda| \,  ,  
  \end{equation} 
where  $\sigma$ is the spectrum of the operator,  $P_{I} $ is the corresponding spectral projection, and $|\cdot|$ denotes a set's length, or `volume', as appropriate.

\item 
 The   {\bf Minami bound} \cite{Min96}:  the probability of there being multiple eigenvalues in a small energy range satisfies
 \begin{eqnarray} \label{minami} 
    \P{\left\{ \rm{card}\{\sigma(H_\Lambda) \cap I \} \ge 2 \right\} } & \le & 
    \mathbb{E}\left[\Tr P_{I}(H) \, (\Tr P_{I}(H) -1)\right] \nonumber  \\[1ex] 
    & \le  &  \frac{\pi^2}{2} \,  \rho_\infty^2\, |I|^2 \, |\Lambda|^2\,  .  
  \end{eqnarray} 
(The statement had a  one dimensional precursor in \cite{Mol81}.) 
 
\end{enumerate} 
These bounds were recently extended  
\cite{Graf_Vaghi07, BeHiSto07} to: 
 \begin{eqnarray} \label{n-minami} 
     \P{\left\{\rm{card}\{\sigma(H_\Lambda) \cap I \} \ge n \right\}}  \ \le \  
\frac{\pi^n}{n!} \,  \rho_\infty^n \, |I|^n \, |\Lambda|^n\,  .  
  \end{eqnarray} 
Furthermore, in a work which was posted at the time of completion of  this article,   the Minami bound was given a new and more transparent derivation and some further extensions~\cite{CGK08}.  

At first glance, one could ask whether \eqref{n-minami} is a special case of a more general valid  bound on  the $n$-point density functions, of the form: 
 \begin{eqnarray} \label{n-point} 
    \P {\left\{ \sigma(H_\Lambda) \cap I_j \neq \emptyset \quad \mbox{for all $j=1, ..., k$} \right\} } \ \stackrel{??}{\le} \  C_n  \ \rho^n_\infty \, 
\prod_{j=1}^k   |I_j|  \, |\Lambda| \,  ,  
\end{eqnarray} 
where $\{I_j \}$ could be arbitrary intervals.   

A bound like~\eqref{n-point}  could be of use, e.g.,  in estimating the probability that for an a-priori specified energy $E$ there is a of multi-state resonance, in the sense that the quantity $|E-\sum_{j=1}^k m_j E_j|$ is small for some  integer sequences $\{m_j\}$.  Questions of this kind  are of relevance  for   the non-linear extension of the Schr\"odinger evolution  which is studied in~\cite{FKS}.

As it turns out, the bound suggested in \eqref{n-point}  does not hold at the generality of the two preceding statements.  As the simple calculation which is demonstrated next shows,  it is already not valid for the $2\times 2$ example.     
Nevertheless, a somewhat similar bound does hold for disjoint energy intervals under the restriction that the eigenfunctions' moduli are of sufficiently different profiles.  
The precise statement, which is our main result, is presented in section~\ref{sect:thebound}.

\section{A counterexample} \label{counterexample} 

While  at first glance \eqref{n-point} may appear sensible, and even  supported by the observation that for random matrices the level interaction is repulsive, it is easily seen to be false.  A counterexample to \eqref{n-point} is found already in the context of $2\times 2$ matrices.

In the two dimensional space, a self adjoint operator with random potential is given by a $ 2 \times 2 $ self adjoint matrix of the form: 
\begin{equation}
  H(\omega) = \left(\begin{matrix} a+\omega_1 & c \\ c^* & b+\omega_2 \end{matrix}\right) \, ,
\end{equation} 
with some $ a, b \in \mathbb{R} $ and $ c \in \mathbb{C} $. The two eigenvalues of $ H(\omega) $ are  
\begin{equation} \label{E_omega}
  E_{1/2}(\omega) = \frac{1}{2} \left( \omega_1 + \omega_2 + a + b \pm \sqrt{(\omega_1 - \omega_2 + a -b)^2 + 4 |c|^2 } \right) \, , 
\end{equation}
where one may note that the spectral gap satisfies: $ |E_1(\omega) - E_2(\omega)| \geq 2 |c| $ for all $ \omega \in \mathbb{R}^2$.\\ 

The determinant of the change of variables $ (\omega_1,\omega_2) \to (E_1,E_2) $ is given by
\begin{align}
  \left|\det\left(\frac{\partial E_j(\omega)}{\partial \omega_k}\right) \right| 
  & = \frac{\left| \omega_1 - \omega_2 + a -b \right|}{\sqrt{(\omega_1 - \omega_2 + a -b)^2 + 4 |c|^2 } }\  = \  
  \frac{\sqrt{(E_1 - E_2 )^2 - 4 |c|^2 }}
  {\left| E_1 - E_2 \right|}
\end{align}
with $j,k \in \{1,2\}$.  
Hence, for $ \omega_1$ and $\omega_2 $ a pair of iid variables with  a common density $ \varrho $, the probability density for the pair of eigenvalues $ p(E_1,E_2) $ 
 (with $ E_{1} \neq E_2$) is:
\begin{eqnarray} \label{counterexample}
  p(E_1,E_2) & =&    \left|\det\left(\frac{\partial E_j(\omega)}{\partial \omega_k}\right) \right|^{-1}  \, \varrho(\omega_j(E_1,E_2))       \\ 
  & =&  \left\{ \begin{array}{cl} \displaystyle \frac{|E_1 - E_2|}{ \sqrt{( E_1 - E_2)^2 - 4 |c|^2} } \, 
    \prod_{j=1}^2 \varrho(\omega_j(E_1,E_2)) \, ,  
    &  | E_1 - E_2| > 2 |c| \\
  0 \, , &  | E_1 - E_2| \leq 2 |c|, \end{array} \right.   \nonumber 
\end{eqnarray} 
with $ \omega_j(E_1,E_2) $ determined by  the relation \eqref{E_omega}.   

The above density has the singularity of $(|E_1-E_2|-2|c|)^{-1/2}$ at the edge of the spectral gap, where $ |E_1-E_2|=2|c|$.    In effect, one could see here that the level repulsion  is associated with a `pile - up' of the probability density at the edge of the gap which it creates.   

Clearly, similar singularities in the two point function would be found in the more general   $n\times n$ situation whenever the systems is decomposable with an isolated two-site subsystem.  However, it seems to be an interesting question whether  for generic $n\times n$ matrices with random potential the singularity is rounded off due to the larger number of random variables.  

One may also note that while the above calculation contradicts \eqref{n-point},  it implies that at least in the $2\times 2$ case the two point function satisfies a modified bound, which is  be obtained by replacing $|I_j|$, on the right side of \eqref{n-point},  with $\max\{|I_j|^{1/2}, |I_j|\}$.   It will be of interest to clarify how far can such a bound be extended.   Suitable generalizations could provide useful information on the probabilities of multi-level resonances which are mentioned above.  

\section {Bounds on the joint distribution of the  eigenvalues}  
\label{sect:thebound}

\subsection{Statement of the main result} 

The `positive results' which are presented here  amount to bounds on the probabilities of simultaneous occurrences of eigenvalues, in prescribed intervals, which are associated with eigenfunctions of sufficiently distinct profiles.   Following is the definition of that concept.

\begin{definition}\label{def:profile} 
Normalized functions $  \psi_1, \dots , \psi_n \in \ell^2(\Lambda) $, with $\| \psi_j \|=1$, 
 are  said to have $\alpha$-distinct profiles, for some $\alpha >0$, 
within sets  $ B_1, \dots , B_n  \subseteq \Lambda $ if and only if  
\begin{equation}\label{eq:detdef}
  \sum_{x_1 \in B_1} \dots \sum_{x_n \in B_n}  
  \left| \det \left( |\psi_j(x_k)|^2 \right)_{j,k=1}^n  \right| 
  \geq \alpha^n \, . 
\end{equation} 
\end{definition}  

It may be noted that by the linearity of the determinant and the triangle inequality:  
\begin{equation} \label{comparison}
 \sum_{x_1 \in B_1} \dots \sum_{x_n \in B_n}  
  \left| \det \left( |\psi_j(x_k)|^2 \right)_{j,k=1}^n  \right| \geq 
\left| \det \left( \langle \psi_j , \indfct_{B_k} \psi_j\rangle \right)_{j,k=1}^n  \right| \, , 
\end{equation}
where $ \indfct_{D} $ stands for the indicator function of the set $ D $.  
Hence, a sufficient condition for~\eqref{eq:detdef} is that the
row (or column) vectors in the (substochastic) matrix of occupation probabilities 
$ \left( \langle \psi_j , \indfct_{B_k} \psi_j\rangle \right)_{j,k=1}^n $
span a parallelepiped of volume at least $ \alpha^n $. \\

We shall now consider the events: 
\begin{align} \label{event}  
& \mathcal{E}_\alpha(I_1. \dots, I_n; B_1, \dots , B_n) \    \\  
 & := \left\{ \omega \, \Big| \, \begin{array}{l}
  \mbox{\rm $ H_\Lambda(\omega) $ has eigenvalues $ E_1 \in I_1 , \dots, E_n \in I_n $ whose  eigen- } \\
    \mbox{\rm functions have  
    $\alpha$-distinct profiles,  within sets  $ B_1, \dots , B_n $,}
  \end{array}\right\}  \nonumber
\end{align}
with $ I_1, \dots, I_n  \subseteq \mathbb{R} $ a collection of Borel sets, and 
$  B_1, \dots , B_n  \subseteq \Lambda $.   \\    

Proven below is the following statement.  

\begin{theorem} \label{main} 
For operators with random potential, as in \eqref{H}, whose probability distribution satisfies the regularity assumption  {\bf R}, the probabilities of the events defined in \eqref{event} satisfy: 
\begin{equation}\label{eq:thm1}
\P{\mathcal{E}_\alpha(I_1. \dots, I_n; B_1, \dots , B_n)  }
  \  \leq \  \frac{ n!}{\alpha^n}  \,  \rho_\infty^n \,  \prod_{j=1}^n |I_j| \, |B_j|   \, .
\end{equation} 
\end{theorem}

Concerning the uses of this result, it may be noted that throughout the localization regime, where the eigenfunctions are each localized in some region of -roughly- the localization length,  condition~\eqref{eq:detdef} would be satisfied by eigenfunctions of separate supports.   Theorem~\ref{main}  can therefore by used to bound the probabilities of eigenfunctions in prescribed intervals whose eigenfunctions do not overlap in space.  Regrettably, the results presented here do not address  the corresponding question for eigenvalues with overlapping eigenfunctions.  One may wonder whether  even in that case the eigenfunctions' profiles  should typically be distinguishable, in the sense of~\eqref{eq:detdef}.  Such a result could extend the applicability of Theorem~\ref{main}.  

\subsection{Multiparameter spectral averaging}

To prove Theorem~\ref{main}, we first establish the following estimate.  

\begin{lemma}[Generalized spectral averaging]\label{lem:specav}
For operators with random potential, as in \eqref{H}, whose probability distribution satisfies the assumption {\bf R}, 
and any collection of intervals $ I_1 , \dots , I_n $ and sites $ x_1, \dots , x_n \in \Lambda $:
\begin{equation}\label{eq:specav}
  \mathbb{E}\left[ \left| \det\left( \left\langle \delta_{x_k} , P_{I_j}(H_\Lambda) \, \delta_{x_k} \right\rangle \right)_{j,k=1}^n \right| \right] 
  \leq n! \; \varrho_\infty^n \, \prod_{j=1}^n |I_j| \, . 
\end{equation}
\end{lemma}
The term by which we refer to this statement is motivated by the observation that the case $n=1$ yields the known spectral averaging principle:  
\begin{equation}
  \mathbb{E}\left[ \left\langle \delta_{x} , P_{I}(H_\Lambda) \, \delta_{x}  \right\rangle \right] 
  \leq \varrho_\infty \, |I| \, ,
\end{equation}
from  which the Wegner estimate \eqref{wegner} readily follows.  
The proof of the more general statement is based on an elementary change of variable calculation, combined with input from algebraic-geometry.  The latter is needed  for a bound (which is rather natural) on a relevant multiplicity factor.

\begin{proof}[Proof of Lemma~\ref{lem:specav}]
We shall first derive \eqref{eq:specav} under the additional restriction  to the event: 
\begin{equation}\label{eq:single}
   \mathcal{J}(I_1, \dots , I_n) := \left\{\omega \, \Big| \, \begin{array}{l} 
  \mbox{The spectrum of $ H_\Lambda(\omega) $ includes exactly one }\\
    \mbox{eigenvalue in each of the intervals $ I_1,.., I_n $} \end{array} \right\} \, .
\end{equation}
For $\omega \in  \mathcal{J}(I_1, \dots , I_n)$  the determinant in the left side of \eqref{eq:specav} reduces to 
\begin{equation}
  D\big(\boldsymbol{E}; \Sigma \big) := \det \mbox{}_{n\times n} (|\psi_j(x_k)|^2)  \, , \qquad \Sigma:=\{x_1 , \dots , x_n\} \, ,
\end{equation} 
where  $\boldsymbol{E} := (E_1, \dots. E_n ) $ is the set of eigenvalues which occur in the indicated intervals, and $ \psi_j$ are the normalized eigenvectors  of $ H_\Lambda(\omega) $ corresponding to the (uniquely defined)  
eigenvalues $ E_j \in I_{j} $.  Thus, our first goal is to establish the bound
\begin{equation}\label{eq:Detbound}
  \mathbb{E}\left[    \; \left|D\big(\boldsymbol{E}; \Sigma\big)  \right|  \; \indfct_{\mathcal{J}_{ I_1,\dots , I_n  } }  \,  \right] 
  \  \leq \  n! \; \varrho_\infty^n \,\prod_{j=1}^n  |I_j| \, ,
\end{equation} 
where $\indfct_{\mathcal{J} } \equiv \indfct_{\mathcal{J} } (\omega)$ denotes the indicator function of  the event $\mathcal{J}  \equiv \mathcal{J}(I_1, \dots , I_n)$.

The expectation value  in \eqref{eq:Detbound}  can be calculated as the average of  the conditional expectation of the same quantity conditioned on $ V_{\Lambda \backslash \Sigma} := \{ V_{x}\}_{x \not\in \Sigma} $, i.e. 
\begin{equation}\label{eq:ac1}
\mathbb{E}\left[  \left|D\big(\boldsymbol{E}; \Sigma\big)  \right|  \indfct_{\mathcal{J}} \right] \ = \     
\int_{\R^{|\Lambda\backslash \Sigma|}} 
\left[ 
\int_{\R^{|\Sigma|}} 
\indfct_{\mathcal{J} }  \; \left|D\big(\boldsymbol{E}; \Sigma\big)  \right|  \  
\mu(d V_{ \Sigma} \, |\, V_{\Lambda\backslash \Sigma})   
\right] \mu(d V_{\Lambda\backslash \Sigma})  \, .
\end{equation}
Thus, under the assumption ${\bf R}$ on the joint distribution of $\{ V_x\}$,  we have: 
 \begin{equation}\label{eq:ac}
\mathbb{E}\left[  \left|D\big(\boldsymbol{E}; \Sigma\big)  \right|  \indfct_{\mathcal{J}} \right]  \  \leq   \ 
\varrho_\infty^n  \, \sup_{ V_{\Lambda \backslash \Sigma} }
      \int_{\mathcal{S} }  \left|  D\big(\boldsymbol{E}; \Sigma\big)  \right|  \, dV_\Sigma  \, . 
\end{equation}                        
where  ${\mathcal{S} } $ is the following subset of  the section of $ \mathcal{J} $ at the specified $V_{\Lambda\backslash \Sigma}$: 
\begin{equation}\label{eq:posdet}
  \mathcal{S} := \left\{ V_\Sigma \, \Big| \, V \equiv (V_\Sigma, V_{\Lambda\backslash \Sigma}) \in \mathcal{J}  \; \mbox{and} \; D\big(\boldsymbol{E}; \Sigma\big) \neq 0 
  \,\right\} .
\end{equation} 
The integral on the right side of \eqref{eq:ac} may be conveniently expressed through the change of  variables 
\begin{equation}\label{eq:coordtrafo}
V_\Sigma := ( V_{x_1} , \dots ,  V_{x_n} ) \longrightarrow \boldsymbol{E} := (E_1, \dots , E_n ) \, , 
\end{equation} 
which is to be understood as performed  at fixed $V_{\Lambda\backslash \Sigma}$. 
Standard perturbation theory \cite{Kato66} implies that the set  
$ \mathcal{J} \subset \R^{|\Lambda|}$ 
is covered by open sets within each of which $E_j$, 
are defined as single-valued analytic functions of   $V_\Sigma $, with derivatives given by the Feynman-Hellmann formula: 
\begin{equation}  \label{FH} 
\frac{\partial E_{j}}{\partial V_{x_k}} = \left| \psi_{j}(x_k)\right|^2 \, . 
\end{equation}
Hence, the Jacobian  for the coordinate change is given by 
\begin{equation} 
\det \left(  \frac{\partial \{E_1, .., E_n \}} 
{\partial \{V_{x_1}, .., V_{x_n} \}} \right) \ = \ D\big(\boldsymbol{E}; \Sigma\big) 
\end{equation}
The section $ \mathcal{S} $ is covered by open sets on which 
$  D\big(\boldsymbol{E}; \Sigma\big) \neq 0 $, for which the transformation \eqref{eq:coordtrafo} is locally bijective.  Globally, the mapping is not $1 - 1$, and     
the correct change of variables formula  is:
\begin{equation}
  \int_{\mathcal{S} }  \left|  D\big(\boldsymbol{E}; \Sigma\big)  \right|  \, dV_\Sigma 
  = \int_{I_1 \times \dots \times I_n} \mkern-30mu N(\boldsymbol{E}; \Sigma) \; d\boldsymbol{E} \, .
\end{equation}
with the multiplicity factor: 
\begin{equation}
  N(\boldsymbol{E}; \Sigma) := \rm{card} \left\{ V_\Sigma \, | \, \mbox{$\boldsymbol{E}$ are eigenvalues of $ H_\Lambda(V_\Sigma, V_{\Lambda\backslash \Sigma})  $ and  $ D\big(\boldsymbol{E}; \Sigma\big) \neq 0 $} \right\}  \, .
\end{equation}

The factor $ N(\boldsymbol{E}; \Sigma) $ counts the number of simultaneous solutions, for $V_{\Sigma}$, of the set of equations (at fixed $V_{\Lambda\backslash \Sigma} $): 
\begin{equation} \label{system}
P_{E_j}(V_{x_1}, \dots, V_{x_n}) \ = \  0  \, , \qquad \mbox{$j=1,..., n$} \, , 
\end{equation}
where $P_{E}(V_{x_1}, \dots, V_{x_n}) $ is the characteristic polynomial: 
\begin{equation}\label{eq:charpol}
  P_{E}(V_{x_1}, \dots, V_{x_n}) \  \equiv \   
  P(V_{x_1}, \dots, V_{x_n}; E) \ :=\  \det\left( H_\Lambda - E \right) 
\end{equation} 
The number of solutions of a system of algebraic equations is a classical problem of algebraic geometry (for which it is the size of a zero dimensional algebraic variety, defined by \eqref{system}).  A rather general answer is provided by  so-called  Bezout's theory.   However, a simplifying observation is that in the case of interest for us the polynomials 
$P_{E}(V_{x_1}, \dots, V_{x_n}) $   are linear in each of the $V_{x_j}$ variables.   
To form a guess as to what may be the number of solutions for such systems, one may observe that the answer is trivial  
if the non-random term  in $H_{\Lambda}(\omega) = T+V(\omega)$  is a diagonal matrix $\rm{diag}\{T_{x_1}, ... ,  T_{x_n} \}$.  In this case, equations \eqref{system} are simultaneously satisfied if and only if  $\{T_{x_j}+ V_{x_j}\}_{j=1,...,n}$ coincide with a permutation of $\{E_j\}_{j=1,...,n}$, and thus  $N(\boldsymbol{E}; \Sigma)  \leq n!$.

As it turns out, by a theorem due to D. Bernstein\footnote{We thank J. Koll\'ar for help with the reference.} also in the more general case which is of interest to us  the number of solutions of the system  \eqref{system} satisfies:       
\begin{equation}  \label{n!}
 N(\boldsymbol{E}; \Sigma)   \  \le \  n!  \quad   .
\end{equation} 
The applicable theorem  is Proposition~\ref{lem:bezout} which is presented in Appendix~\ref{App:counting}.   To apply it, we need to check that for the counted solutions $\det\left(  \frac{\partial P_{E_j}}{\partial V_{x_k} }   \right)   \neq 0$.    
For that we note: 
\begin{equation}
  \frac{\partial P(V_\Sigma; E_j)}{\partial V_{x_k} } 
  \ = \ \frac{ \partial \mbox{  }} {\partial V_{x_k}} {\big |_{E=E_j}}
 \prod_m \left [ E_m(V_\Sigma)-E \right ]  
  =   \frac{\partial E_{j}}{\partial V_{x_k}}  \, \prod_{\substack{E_{m} \in \sigma(H_\Lambda) \\ m \neq j}} ( E_m - E_j) \, .  
\end{equation} \\  
Since the last product is non-zero for  all $ V \in \mathcal{J} $ and $ j = 1 , \dots , n $, condition \eqref{eq:detcond}
is satisfied on $ \mathcal{S} $.\\


The above considerations prove   \eqref{eq:Detbound}, which differs from \eqref{eq:specav} mainly in the presence of the additional constraint that each interval $I_j$ includes exactly one eigenvalue of $H_\Lambda$. 
We shall now show that \eqref{eq:specav} follows. 
For that, consider a  partition of $\cup_j I_j $ into a finite collection, ${\mathcal C_\varepsilon}$, of disjoint sub intervals whose length does not exceed $\varepsilon $.   One may represent each of the intervals  $I_j$ as a disjoint union  $I_j = \cup_{m=1}^{\ell_j(\varepsilon)}  I_j^{m}$ of elements  of such a partition, i.e. with each $ \{I_j^{m}\}$  in ${\mathcal C_\varepsilon}$  (if the sets $I_j$ are not pairwise disjoint some elements of ${\mathcal C_\varepsilon}$ will be called upon more than once.) 
%
By the linearity of the determinant: 
\begin{equation} \label{sum1}
 \det\left( \left\langle \delta_{x_k} , P_{I_j}(H_\Lambda)\,  \delta_{x_k} \right\rangle \right) = 
 \sum_{m_1=1}^{\ell_1(\varepsilon)}\  \dots \  \sum_{m_n=1}^{\ell_n(\varepsilon)}
  \det\left( \left\langle \delta_{x_k} , P_{ I_{j}^{m_j}}(H_\Lambda) \, \delta_{x_k} \right\rangle \right) \, ,
\end{equation}
The sum can be restricted to the case that $ I_{1}^{_{m_1}} , \dots , I_{n}^{_{m_n}} $ are disjoint, since otherwise the determinant vanishes. 

Using the fact that the determinant on the left is  bounded by one, we can estimate: 
\begin{multline}
  \mathbb{E}\left[ \left|\det\left( \left\langle \delta_{x_k} , P_{I_j}(H_\Lambda)\,  \delta_{x_k} \right\rangle \right)\right| \right]
   \leq   \mathbb{P}\left\{\begin{array}{l} \mbox{In at least one of the elements of ${\mathcal C_\varepsilon}$ }\\ 
  \mbox {$H_\Lambda$ has two or more eigenvalues} \end{array}\right\} \\
  + \sum_{m_1, \dots , m_n} \mathbb{E}\left[\left|\det\left(\left\langle \delta_{x_k} , P_{ I_{j}^{m_j}}(H_\Lambda) \, \delta_{x_k}\right) \right\rangle\right| \; 
    \indfct_{\mathcal{J}(I_{1}^{_{m_1}} , \dots , I_{n}^{_{m_n}})} \right]\, . 
\end{multline}
where the summation range is as in \eqref{sum1}.   We shall now take the limit  $ \varepsilon   \to 0 $.  Since the eigenvalues of $ H_\Lambda $ 
are almost surely simple (Lemma~\ref{lem:simplicity}), 
the dominated convergence theorem implies that 
in  that limit the first term vanishes .  
Applying  \eqref{eq:Detbound} to the remaining terms one gets: 
\begin{align} \label{sum2}
 \E{|\det\left( \left\langle \delta_{x_k} , P_{I_j}(H_\Lambda)\,  \delta_{x_k} \right\rangle \right) |}   \ \le \ &
 \lim_{\varepsilon \to 0}  
 \sum_{m_1=1}^{\ell_1(\varepsilon)}\  \dots \  \sum_{m_n=1}^{\ell_n(\varepsilon)}
 n! \, \rho_\infty^n \prod_j |I_j^{m_j} | \notag \\ \equiv \ & n! \, \rho_\infty^n \prod_j |I_j |
\end{align}
which yields  \eqref{eq:specav}. 
  \end{proof}


\subsection{Proof of main result}     
    
\begin{proof}[Proof of Theorem~\ref{main}]
Using the limiting argument employed at the end of the proof of Lemma~\ref{lem:specav}, one shows that it is sufficient to bound the probability of 
$ \mathcal{E}_\alpha\cap  \mathcal{J} $, where $ \mathcal{J} $ was defined in \eqref{eq:single}. Using the assumption
on the normalized eigenfunctions $ \psi_1, \dots , \psi_n $ corresponding to $ E_1 \in I_1 , \dots , E_n \in I_n $, we then estimate:
\begin{align}
& \mathbb{P}\left(\mathcal{E}_\alpha(I_1. \dots, I_n; B_1, \dots , B_n) \cap  \mathcal{J}(I_1. \dots, I_n)\right) \notag \\
& \leq  \frac{1}{\alpha^n}  \sum_{x_1 \in B_1} \dots \sum_{x_n \in B_n} 
  \mathbb{E}\left[ \left| \det \left( |\psi_{j}(x_k)|^2 \right)  \right|  \, 
    1_{ \mathcal{J}(I_1. \dots, I_n) }  \right] \notag \\ 
& \leq \alpha^{-n} \, n! \, \varrho_\infty^n \,  \prod_{j=1}^n |I_j| \, |B_j| 
\end{align}
where the last inequality is due to \eqref{eq:Detbound}.
\end{proof}   


\vskip 2cm 

\appendix

\noindent {\Large \bf  Appendix} \mbox{ } \\ 

\section{Counting solutions of a system of polynomial equations} \label{App:counting}

In the proof of Lemma~\ref{lem:specav},  for the bound \eqref{n!} we invoked the following general result.
\begin{proposition}[Special case of a theorem by D. Bernstein] \label{lem:bezout}
Let $P_J$, $j=1,..., n$ be polynomials in $n$ variables,  
$ \boldsymbol{\sigma} =  (\sigma_1 ,\cdots , \sigma_n) \in \mathbb{C}^n $,  which are linear in each variable, i.e., 
are of the form
\begin{equation}\label{eq:bernsteinpol}
  P_j(\boldsymbol{\sigma}) = \sum_{\boldsymbol{k} \in \{0,1\}^n } c_j(\boldsymbol{k}) \; \sigma_1^{k_1} \cdots \sigma_k^{k_n}\, , 
\end{equation}
with $ c_j(\boldsymbol{k}) \in \mathbb{C} $ which are non-zero only if $k_m= 0,1$. Then the number of isolated solutions of the system
\begin{equation}
  P_j(\boldsymbol{\sigma}) = 0 \, , \qquad \mbox{for all $ j \in \{ 1, \dots , n \}$,}
\end{equation} 
is at most $ n! $. Moreover, each solution $ \boldsymbol{\sigma} $  at which
\begin{equation}\label{eq:detcond}
  \det\left( \frac{\partial P_j(\boldsymbol{\sigma})}{\partial \sigma_k} \right)_{j,k=1}^n \neq 0 \, ,
\end{equation} 
is isolated.
\end{proposition}
 
The first part is a special case of \cite[Thm.~5.4]{CLS91}.
In case the solution is not isolated, there exists locally a differentiable curve $ s \mapsto   \boldsymbol{\sigma}(s) $ such that for all $  j = 1, \dots , n $
\begin{equation}
  0 = \frac{d P_j(\boldsymbol{\sigma}(s)) }{ d s } = \sum_{k=1}^n  \frac{\partial P_j(\boldsymbol{\sigma}(s))}{\partial \sigma_k} \frac{ d \sigma_k(s) }{d s } \, .
\end{equation}
This contradicts the assumption \eqref{eq:detcond}, which implies that the matrix of partial derivatives has no zero eigenvalue.

\section{Simplicity of the spectrum}  

In our discussion it was convenient to know that the spectrum of an operator with random potential is almost surely non-degenerate.  While this assertion is among the consequence of the Minami bound, for completeness we present here also a direct and elementary proof.  

\begin{lemma}\label{lem:simplicity}  Let $ H_\Lambda(\omega)$ be an operator  in $\ell(\Lambda)^2$, for some finite region  $|\Lambda|$, with a random potential such that for each $x\in \Lambda$ the conditional distribution of $V_x(\omega)$, conditioned on $\{V_y(\omega)\}_{y\in \{ x\}^c}$, is almost surely continuous.  Then
for almost all $ \omega $ the spectrum of $ H_\Lambda(\omega) $ has only simple eigenvalues.
\end{lemma}
\begin{proof}[Proof of Lemma~\ref{lem:simplicity}]
Let $ \psi_1 , \dots, \psi_{|\Lambda|} $ be an orthonormal basis of eigenfunctions of $ H_\Lambda $ with corresponding eigenvalues denoted by $ E_1, \dots, E_{|\Lambda|} $.  
Consider the self adjoint operator 
$$  M_\Lambda := ( H_\Lambda \otimes 1 - 1 \otimes H_\Lambda )^2 
$$ 
on the subspace $ \mathcal{H}^-$ of antisymmetric functions within the product space $ \ell^2(\Lambda) \otimes  \ell^2(\Lambda) $. It is straightforward to check that  the orthonormal basis given by 
\begin{equation}
\Psi_{jk}^- := \frac{1}{\sqrt{2}} \left( \psi_j \otimes \psi_k - \psi_k \otimes \psi_j \right) \, , \qquad  j < k \, ,
\end{equation}
constitutes an eigenbasis with $
M_\Lambda \, \Psi_{jk}^- = (E_j - E_k)^2 \,   \Psi_{jk}^- $.
The simplicity of the spectrum of $ H_\Lambda $ is therefore equivalent to 
 $ M_\Lambda $ being almost surely invertible on $ \mathcal{H}^-$, i.e.,
\begin{equation}\label{eq:detnull}
\det M_\Lambda > 0 \, . 
\end{equation}
For a proof of this assertion, we consider the $ \left( |\Lambda| \atop 2 \right)^2 $ matrix-elements given by
\begin{equation} 
  \langle \delta_{x'y'}^- , M_\Lambda \delta_{xy}^- \rangle, \qquad  \delta_{xy}^- := \frac{1}{\sqrt{2}} \left( \delta_x \otimes \delta_y - \delta_y \otimes \delta_x \right)
\end{equation} 
associated
with the antisymmetrized position basis of $ \mathcal{H}^- $, and study their 
dependence on a single random variable, say $ V_x $. 
Only $ |\Lambda |-1 $ rows (and columns) of the matrix depend on $ V_x $.
In these rows, the diagonal matrix elements, with $ x = x' $ and $ y = y' $, depend on $ V_x $ quadratically, while the diagonal elements are  linear in $V_x$.

Hence, $ V_x \mapsto \det M_\Lambda $ is a polynomial of degree at most $ |\Lambda| $.  Thus, we have the following dichotomy:  the characteristic polynomial has either at most $  |\Lambda| $ isolated zeros or is independent of $ V_x $.  It the first case, the conditional probability
that $ \det M_\Lambda = 0 $, conditioned on $ \{ V_y \}_{y \neq x } $, vanishes, since the distribution of $V_x$ is assumed to be continuous with respect to Lebesgue measure.
In the second case, one may reduce the site $ x $ from $ \Lambda $ by taking the limit $ V_x \to \infty $.  In this limit 
$ H_\Lambda \to  H_{\Lambda \backslash \{x\} }\oplus \infty $. Since $  \det M_\Lambda $ does not diverge in this limit, one may conclude that $ \det M_{\Lambda \backslash \{0 \} } = 0 $, 
and the argument may be  repeated  for the smaller set.   In case $ |\Lambda | = 1 $ the condition \eqref{eq:detnull} is trivially satisfied.   This completes the proof.
\end{proof}

\section*{Acknowledgement}
We  thank Shmuel Fishman, Avi Soffer, and Yevgeny Krivolapov for discussions of their work which have stimulated our interest in the topic presented here, and the Weizmann Institute Center for Complex Systems for its hospitality, on a visit for which partial travel support was received from the BSF grant 710021. 
The work was supported in parts by the NSF grants DMS-0602360 (MA) and DMS-0701181 (SW).


\bibliographystyle{plain}

\end{document}